\documentclass[11pt]{article}

\pdfoutput=1

\usepackage{geometry}
\usepackage{amsmath, amsthm, amssymb}

\usepackage{microtype}
\usepackage{fullpage}

\usepackage[backref=page]{hyperref}
\hypersetup{
	bookmarks=true,
	bookmarksnumbered=true,
	bookmarksopen=true,
	pdftitle={The simplex method is strongly polynomial for deterministic Markov decision processes},
	pdfauthor={Ian Post, Yinyu Ye}
}

\newtheorem{thm}{Theorem}[section]
\newtheorem{lem}[thm]{Lemma}

\newtheorem{defn}[thm]{Definition}

\newcommand{\mb}[1]{ \textbf{#1} }
\newcommand{\mbv}{\mb{v}}
\newcommand{\mbx}{\mb{x}}
\newcommand{\mbr}{\mb{r}}

\DeclareMathOperator{\argmax}{\textrm{argmax}}

\begin{document}

\title{The simplex method is strongly polynomial for deterministic Markov decision processes}

\author{Ian Post
\thanks{Department of Combinatorics and Optimization, University of Waterloo. Research done while at Stanford University.
Email: \href{mailto:ian@ianpost.org}{ian@ianpost.org}. 
Research supported by NSF grant 0904325. We also acknowledge financial support from grant \#FA9550-12-1-0411 from the U.S. Air Force Office of Scientific Research (AFOSR) and the
Defense Advanced Research Projects Agency (DARPA).
}
\and Yinyu Ye
\thanks{Department of Management Science and Engineering, Stanford University.
Email: \href{mailto:yinyu-ye@stanford.edu}{yinyu-ye@stanford.edu}.}
}

\maketitle

\pdfbookmark[1]{Abstract}{MyAbstract}
\begin{abstract}
We prove that the simplex method with the highest gain/most-negative-reduced cost pivoting rule converges in strongly polynomial time for deterministic Markov decision processes (MDPs) regardless of the discount factor. For a deterministic MDP with $n$ states and $m$ actions, we prove the simplex method runs in $O(n^3 m^2 \log^2 n)$ iterations if the discount factor is uniform and $O(n^5 m^3 \log^2 n)$ iterations if each action has a distinct discount factor. Previously the simplex method was known to run in polynomial time only for discounted MDPs where the discount was bounded away from 1 \cite{ye_simplex}. 

Unlike in the discounted case, the algorithm does not greedily converge to the optimum, and we require a more complex measure of progress. We identify a set of layers in which the values of primal variables must lie and show that the simplex method always makes progress optimizing one layer, and when the upper layer is updated the algorithm makes a substantial amount of progress. In the case of nonuniform discounts, we define a polynomial number of ``milestone'' policies and we prove that, while the objective function may not improve substantially overall, the value of at least one dual variable is always making progress towards some milestone, and the algorithm will reach the next milestone in a polynomial number of steps.
\end{abstract}

\section{Introduction}

Markov decision processes (MDPs) are a powerful tool for modeling repeated decision making in stochastic, dynamic environments. An MDP consists of a set of states and a set of actions that one may perform in each state. Based on an agent's actions it receives rewards and affects the future evolution of the process, and the agent attempts to maximize its rewards  over time (see Section \ref{section_prelim} for a formal definition). MDPs are widely used in machine learning, robotics and control, operations research, economics, and related fields. See the books \cite{puterman} and \cite{bertsekas} for a thorough overview.

Solving MDPs is also an important problem theoretically. Optimizing an MDP can be formulated as a linear program (LP), and although these LPs possess extra structure that can be exploited by algorithms like Howard's policy iteration method \cite{howard}, they lie just beyond the point at which our ability to solve LPs in strongly-polynomial time ends (and are a natural target for extending this ability), and they have proven to be hard in general for algorithms previously thought to be quite powerful, such as randomized simplex pivoting rules \cite{fhz_randsimplex}.

In practice \cite{littmandeankaelbling} MDPs are solved using policy iteration, which may be viewed as a parallel version of the simplex method with multiple simultaneous pivots, or value iteration \cite{bellman}, an inexact approximation to policy iteration that is faster per iteration. If the discount factor $\gamma$, which determines the effective time horizon (see Section \ref{section_prelim}), is small it has long been known that policy and value iteration will find an $\epsilon$-approximation to the optimum \cite{bellman}. It is also well-known that value iteration may be exponential, but policy iteration resisted worst-case analysis for many years. It was conjectured to be strongly polynomial but except for highly-restricted examples \cite{madani} only exponential time bounds were known \cite{mansoursingh}.  Building on results for parity games \cite{paritygame}, Fearnley recently gave an exponential lower bound \cite{fearnley}. Friedmann, Hansen, and Zwick extended Fearnley's techniques to achieve sub-exponential lower bounds for randomized simplex pivoting rules \cite{fhz_randsimplex} using MDPs, and Friedmann gave an exponential lower bound for MDPs using the least-entered pivoting rule \cite{friedmann_simplex}.
Melekopoglou and Condon proved several other simplex pivoting rules are exponential \cite{melecondon_simplex}.

On the positive side, Ye designed a specialized interior-point method that is strongly polynomial in everything except the discount factor \cite{ye_intpt}.
Ye later proved that for discounted MDPs with $n$ states and $m$ actions, the simplex method with the most-negative-reduced-cost pivoting rule and, by extension, policy iteration, run in time $O(nm/(1-\gamma)\log(n/(1-\gamma)))$ on discounted MDPs, which is polynomial for fixed $\gamma$ \cite{ye_simplex}. Hansen, Miltersen, and Zwick improved the policy iteration bound to $O(m/(1-\gamma)\log(n/(1-\gamma)))$ and extended it to both value iteration as well as the strategy iteration algorithm for two player turn-based stochastic games \cite{strategyiteration}.

But the performance of policy iteration and simplex-style basis-exchange algorithms on MDPs remains poorly understood. Policy iteration, for instance, is conjectured to run in $O(m)$ iterations on deterministic MDPs, but the best upper bounds are exponential, although a lower bound of $O(m)$ is known \cite{hansenzwick_mmc}.
Improving our understanding of these algorithms this is an important step in designing better ones with polynomial or even strongly-polynomial guarantees. 

Motivated by these questions, we analyze the simplex method with the most-negative-reduced-cost pivoting rule on deterministic MDPs. For a deterministic MDP with $n$ states and $m$ actions, we prove that the simplex method terminates in $O(n^3m^2\log^2 n)$ iterations regardless of the discount factor, and if each action has a distinct discount factor, then the algorithm runs in $O(n^5 m^3 \log^2 n)$ iterations. Our results do not extend to policy iteration, and we leave this as a challenging open question.

Deterministic MDPs were previously known to be solvable in strongly polynomial time using specialized methods not applicable to general MDPs---minimum mean cycle algorithms \cite{papadimitriou_mdp} or, in the case of nonuniform discounts, by exploiting the property that the dual LP has only two variables per inequality \cite{hochbaum_tvpi}. The fastest known algorithm for uniformly discounted deterministic MDPs runs in time $O(mn)$ \cite{madani2010discounted}.
However, these problems were not known to be solvable in polynomial time with the more-generic simplex method. More generally, we believe that our results help shed some light on how algorithms like simplex and policy iteration function on MDPs.

Our proof techniques, particularly in the case of nonuniform discounts, may be of independent interest. For uniformly discounted MDPs, we show that the values of the primal flux variables must lie within one of two intervals or layers of polynomial size depending on whether an action is on a path or a cycle. Most iterations update variables in the smaller path layer, and we show these converge rapidly to a locally optimal policy for the paths, at which point the algorithm must update the larger cycle layer and makes a large amount of progress towards the optimum. Progress takes the form of many small improvements interspersed with a few much larger ones rather than uniform convergence.

The nonuniform case is harder, and our measure of progress is unusual and, to the best of our knowledge, novel. We again define a set of intervals in which the value of variables on cycles must fall, and these define a collection of intermediate milestone or checkpoint values for each dual variable (the value of a state in the MDP). Whenever a variable enters a cycle layer, we argue that a corresponding dual variable is making progress towards the layer's milestone and will pass this value after enough updates.  When each of these checkpoints have been passed, the algorithm must have reached the optimum. We believe some of these ideas may prove useful in other problems as well.

In Section \ref{section_prelim} we formally define MDPs and describe a number of well-known properties that we require. In Section \ref{section_uniform} we analyze the case of a uniform discount factor, and in Section \ref{section_nonuniform} we extend these results to the nonuniform case.

\section{Preliminaries}

\label{section_prelim}

%In this section we formally define the problem and review some well-known properties of MDPs. 
Many variations and extensions of MDPs have been defined, but we will study the following problem.
A Markov decision process consists of a set of $n$ states $S$ and $m$ actions $A$. Each action $a$ is associated with a single state $s$ in which it can be performed, a reward $\mbr_a \in \mathbb{R}$ for performing the action, and a probability distribution $P_a$ over states to which the process will transition when using action $a$. We denote by $P_{a,s}$ the probability of transitioning to state $s$ when taking action $a$. There is at least one action usable in each state.
Let $\mbr$ be the vector of rewards indexed by $a$ with entries $\mbr_a$, $A_s \subset A$ be the set of actions performable in state $s$, and $P$ be the $n$ by $m$ matrix with columns $P_a$ and entries $P_{a,s}$.
We will restrict the distributions $P_a$ to be deterministic for all actions, in which case states may be thought of as nodes in a graph and actions as directed edges. However, the results in this section apply to MDPs with stochastic transitions as well.

At each time step, the MDP starts in some state $s$ and performs an action $a$ admissible in state $s$, at which point it receives the reward $\mbr_a$ and transitions to a new state $s'$ according to the probability distribution $P_a$.
We are given a discount factor $\gamma < 1$ as part of the input, and our goal is to choose actions to perform so as to maximize the expected discounted reward we accumulate over an infinite time horizon. 
%The discount $\gamma$ may be extremely close to 1. 
The discount can be thought of as a stopping probability---at each time step the process ends with probability $1-\gamma$.
Normally, the discount $\gamma$ is uniform for the entire MDP, but in Section \ref{section_nonuniform} we will allow each action to have a distinct discount $\gamma_a$.

Due to the Markov property---transitions depend only the current state and action---there is an optimal strategy that is memoryless and depends only on the current state. Let $\pi$ be such a {\em policy}, a distribution of actions to perform in each state. This defines a Markov chain and a value for each state:

\begin{defn}
\label{def_value}
Let $\pi$ be a policy, $P^\pi$ be the $n$ by $n$ matrix where $P^\pi_{s,s'}$ is the probability of transitioning from $s'$ to $s$ using $\pi$, and $\mbr_\pi$ the vector of expected rewards for each state according to the distribution of actions in $\pi$.
The {\em value vector} $\mbv^\pi$ is indexed by states, and $\mbv^\pi_s$ is equal to the expected total discounted reward of starting in state $s$ and following policy $\pi$. It is defined as $\mbv^\pi = \sum_{i\ge 0} (\gamma (P^\pi)^T)^i\mbr_\pi = (I-\gamma P^\pi)^{-T}\mbr_\pi$ or equivalently by
\begin{equation}
\label{eq_value_vector}
\mbv^\pi = \mbr_\pi + \gamma (P^\pi)^T \mbv^\pi .
\end{equation}
\end{defn}

If policy $\pi$ is randomized and uses two or more actions in some state $s$, then the value of $\mbv^\pi_s$ is an average of the values of performing each of the pure actions in $s$, and one of these is the largest. Therefore we can replace the distribution by a single action and only increase the value of the state. In the remainder of the paper we will restrict ourselves to pure policies in which a single action is taken in each state.

In addition to the value vector, a policy $\pi$ also has an associated flux vector $\mbx^\pi$ that will play a critical role in our analysis. It acts as a kind of ``discounted flow.''  Suppose we start with a single unit of ``mass'' on every state and then run the Markov chain. At each time step we remove $1-\gamma$ fraction of the mass on each state and redistribute the remaining mass according to the policy $\pi$. Summing over all time steps, the total amount of mass that passes through each action is its flux.  More formally,

\begin{defn}
\label{def_flux}
Let $\pi$ be a policy and $P^\pi$ the $n$ by $n$  transition matrix for $\pi$ formed by the columns $P_a$ for actions in $\pi$. The {\em flux vector} $\mbx^\pi$ is indexed by actions. If action $a$ is not in $\pi$ then $\mbx^\pi_a = 0$, and if $\pi$ uses $a$ in state $s$, then $\mbx^\pi_a = \mb{y}_s$, where
\begin{equation}
\label{eq_flux_vector}
\mb{y} = \sum_{i \ge 0} (\gamma P^\pi)^i \mb{1} = (I - \gamma P^\pi)^{-1} \mb{1} \; ,
\end{equation}
and $\mb{1}$ is the all ones vector of dimension $n$. The flux is the total discounted number of times we use each action if we start the MDP in all states and run the Markov chain $P^\pi$ discounting by $\gamma$ each iteration.
\end{defn}

Note that if $a \in \pi$ then $\mbx^\pi_a \ge 1$, since the initial flux placed on $a$'s state always passes through $a$. Further note that each bit of flux can be traced back to one of the initial units of mass placed on each state, although the vector $\mbx^\pi$ sums flux from all states. This will be important in Section \ref{section_nonuniform}.

Solving the MDP can be formulated as the following primal/dual pair of LPs, in which the flux and value vectors correspond to primal and (possibly infeasible) dual solutions:

\begin{equation}
\label{eq_primal}
\begin{array}{rlrl}
\textsc{Primal:} && \\
\textrm{maximize} & \multicolumn{2}{l}{ \sum_a \mbr_a\mbx_a } \\
\textrm{subject to} &\forall s \in S, & \sum_{a \in A_s} \mbx_a & = 1 + \gamma \sum_{a} P_{a,s}\mbx_a\\
& & \mbx & \ge 0 \\
\end{array}
\end{equation}

\begin{equation}
\label{eq_dual}
\begin{array}{rlrl}
\textsc{Dual:} && \\
\textrm{minimize} & \multicolumn{2}{l}{ \sum_s \mbv_s } \\
\textrm{subject to} &\forall s \in S, a \in A_s, & \mbv_s \ge \mbr_a + \gamma \sum_{s'} P_{a,s'}\mbv_{s'} \\
\end{array}
\end{equation}

The constraint matrix of \eqref{eq_primal} is equal to $M - \gamma P$, where $M_{s,a} = 1$ if action $a$ can be used in state $s$ and 0 otherwise.
The dual value LP \eqref{eq_dual} is often defined as the primal, as it is perhaps more intuitive, and \eqref{eq_primal} is rarely considered. However, our analysis centers on the flux variables, and algorithms that manipulate policies can more naturally be seen as moving through the polytope \eqref{eq_primal}, since vertices of the polytope represent policies:

\begin{lem}
The LP \eqref{eq_primal} is non-degenerate, and there is a bijection between vertices of the polytope and policies of the MDP.
\label{lem_vertices}
\end{lem}

\begin{proof}
Policies have exactly $n$ nonzero variables, and solving for the flux vector in \eqref{eq_flux_vector} is identical to solving for a basis in the polytope, so policies map to bases.
Write the constraints in \eqref{eq_primal} in the standard matrix form $A\mbx = \mb{b}$.
The vector $\mb{b}$ is $ \mb{1}$, and $A = M - \gamma P$. In a row $s$ of $A$ the only positive entries are on actions usable in state $s$, so if $A\mbx = \mb{b}$, then $\mbx$ must have a nonzero entry for every state, i.e., a choice of action for every state. Bases of the LP have $n$ variables, so they must include only one action per state.

Finally, as shown above $\mbx^\pi_a \ge 1$ for all $a$ in a policy/basis, so the LP is not degenerate, and bases correspond to vertices.
\end{proof}

By Lemma \ref{lem_vertices}, the simplex method applied to \eqref{eq_primal} corresponds to a simple, single-switch version of policy iteration: we start with an arbitrary policy, and in each iteration we change a single action that improves the value of some state.
Since the LP is not degenerate, the simplex method will find the optimal policy with no cycling.
We will use Dantzig's most-negative-reduced-cost pivoting rule to choose the action switched. Since \eqref{eq_primal} is written as a maximization problem, we will refer to reduced costs as gains and always choose the highest gain action to switch/pivot. For MDPs, the gains have a simple interpretation:

\begin{defn}
\label{def_reduced_costs}
The {\em gain} (or {\em reduced cost}) of an action $a$ for state $s$ with respect to a policy $\pi$ is denoted $\mbr^\pi_a$ and is the improvement in the value of $s$ if $s$ uses action $a$ once and then follows $\pi$ for all time. Formally, $\mbr^\pi_a = (\mbr_a + \gamma P_a^T \mbv^{\pi}) - \mbv^\pi_s$, or, in vector form
\begin{equation}
\label{eq_reducedcost}
\mbr^\pi = \mbr - (M - \gamma P)^T\mbv^\pi \; .
\end{equation}
\end{defn}

We denote the optimal policy by $\pi^*$, and the optimal flux, values, and gains by $\mbx^*$, $\mbv^*$, and $\mbr^*$.
The following are basic properties of the simplex method, and we prove them for completeness.

\begin{lem}
\label{lem_reducedcosts}
Let $\pi$ and $\pi'$ be any policies. The gains satisfy the following properties
\begin{itemize}

\item $(\mb{r}^\pi)^T\mb{x}^{\pi'} = \mb{r}^T\mb{x}^{\pi'} - \mb{r}^T\mb{x}^\pi = \mb{1}^T\mb{v}^{\pi'}-\mb{1}^T\mb{v}^\pi$,

\item $\mbr^\pi_a = 0$ for all $a \in \pi$, and

\item $\mbr^*_a \le 0$ for all $a$.

\end{itemize}
\end{lem}

\begin{proof}
From the definition of gains $(\mbr^\pi)^T\mbx^{\pi'} = (\mbr - (M - \gamma P)^T\mbv^\pi)^T\mbx^{\pi'} = \mbr^T\mbx^{\pi'} - (\mbv^\pi)^T(M-\gamma P) \mbx^{\pi'} = \mbr^T\mbx^{\pi'} - (\mbv^{\pi})^T\mb{1}$, using that $(M-\gamma P)$ is the constraint matrix of \eqref{eq_primal}. 
From the definition of value and flux vectors $\mbr^T\mbx^{\pi} = \mbr_\pi^T (I-\gamma P^\pi)^{-1}\mb{1} = (\mbv^\pi)^T \mb{1}$, where $\mbr_\pi$ is the reward vector restricted to indices $\pi$. Combining these two gives the first result.

For the second result, if $a$ is in $\pi$, then $\mbv^\pi_s = \mbr_a + \gamma P_a^T \mbv^\pi$, so $\mbr^\pi_a = 0$.
Finally, if $\mbr^*_a > 0$ for some $a$, then consider the policy $\pi$ that is identical to $\pi^*$ but uses $a$. Then $(\mbr^*)^T\mbx^{\pi} > 0$, and the first identity proves that $\pi^*$ is not optimal.
\end{proof}

A key property of the simplex method on MDPs that we will employ repeatedly is that not only is the overall objective improving, but also the values of all states are monotone non-decreasing, and there exists a single policy we denote by $\pi^*$ that maximizes the values of all states:

\begin{lem}
\label{lem_monotone}
Let $\pi$ and $\pi'$ be policies appearing in an execution of the simplex method with $\pi'$ being used after $\pi$. Then $\mbv^{\pi'} \ge \mbv^{\pi}$. Further, let $\pi^*$ be the policy when simplex terminates, and $\pi''$ be any other policy. Then $\mbv^* \ge \mbv^{\pi''}$.
\end{lem}

\begin{proof}
Suppose $\pi$ and $\pi'$ are subsequent policies.
The gains of all actions in $\pi'$ with respect to $\pi$ are equal to $\mbr_{\pi'} - (I - \gamma P^{\pi'})^T\mbv^{\pi}$, all of which are nonnegative. Therefore $\mb{0} \le (I-\gamma P^{\pi'})^{-T}(\mbr_{\pi'} - (I - \gamma P^{\pi'})^T)\mbv^{\pi} = \mbv^{\pi'} - \mbv^\pi$, using that $(I-\gamma P^{\pi'})^{-T} = \sum_{i\ge 0} (\gamma (P^\pi)^T)^i \ge \mb{0}$. By induction, this holds if $\pi$ and $\pi'$ occur further apart.
Performing a similar calculation using the gains $\mbr^*$, which are nonpositive, shows that $\mbv^* - \mbv^{\pi''} \ge \mb{0}$ for any policy $\pi''$.
\end{proof}

\section{Uniform discount}

\label{section_uniform}

As a warmup before delving into our analysis of deterministic MDPs, we briefly review the analysis of \cite{ye_simplex} for stochastic MDPs with a fixed discount. Consider the flux vector in Definition \ref{def_flux}. One unit of flux is added to each state, and every step it is discounted by a factor of $\gamma$, for a total of $n(1 + \gamma + \gamma^2 + \cdots) = n/(1-\gamma)$ flux overall. If $\pi$ is the current policy and $\Delta$ is the highest gain, then, by Lemma \ref{lem_reducedcosts} the farthest $\pi$ can be from $\pi^*$ is if all $n/(1-\gamma)$ units of flux in $\pi^*$ are on the action with gain $\Delta$, so $\mbr^T\mbx^* - \mbr^T\mbx^\pi \le n\Delta/(1-\gamma)$. If we pivot on this action, at least 1 unit of flux is placed on the new action, increasing the objective by at least $\Delta$. Thus we have reduced the gap to $\pi^*$ by a $1 - (1-\gamma)/n$ fraction, which is substantial if $1/(1-\gamma)$ is polynomial. 

Now consider $\mbr^T\mbx^* - \mbr^T\mbx^\pi  = -(\mbr^*)^T\mbx^\pi$. All the terms $-\mbr^*_a \mbx^\pi_a$ are nonnegative, and for some action $a$ in $\pi$ we have $-\mbr^*_a \mbx^\pi_a \ge -(\mbr^*)^T\mbx^\pi / n$. The term $-\mbr^*_a \mbx^\pi_a$ is at most $-\mbr^*_a n/(1-\gamma)$, so $-\mbr^*_a \ge -(\mbr^*)^T\mbx^\pi / (n^2/(1-\gamma))$. But for any policy $\pi'$ that includes $a$, $-(\mbr^*)^T\mbx^{\pi'} \ge -\mbr^*_a \mbx^{\pi'}_a \ge - \mbr^*_a$, so after $\mbr^T\mbx^* - \mbr^T\mbx^\pi$ has shrunk by a factor of $n^2/(1-\gamma)$, action $a$ cannot appear in any future policy, and this occurs after
\[
\log_{1-(1-\gamma)/n} \frac{1-\gamma}{n^2} = O\left( \frac{n}{1-\gamma}\log \frac{n}{1-\gamma} \right)
\]
steps.
%When the gap has been reduced sufficiently, we are able to eliminate an action from all future policies. 
See \cite{ye_simplex} for the details.

The above result hinged on the fact that the size of all nonzero flux lay within the interval $[1,n/(1-\gamma)]$, which was assumed to be polynomial but gives a weak bound if $\gamma $ is very close to 1. However, consider a policy for a deterministic MDP. It can be seen as a graph with a node for each state with a single directed edge leaving each state representing the action, so the graph consists of one or more directed cycles and directed paths leading to these cycles. Starting on a path, the MDP uses each path action once before reaching a cycle, so the flux on paths must be small. Flux on the cycles may be substantially larger, but since the MDP revisits each action after at most $n$ steps, the flux on cycle actions varies by at most a factor of $n$.

\begin{lem}
\label{lem_flux_size}
Let $\pi$ be a policy with flux vector $\mbx^\pi$ and $a$ an action in $\pi$. If $a$ is on a path in $\pi$ then $1 \le \mbx^\pi_a \le n$, and if $a$ is on a cycle then $1/(1-\gamma) \le \mbx^\pi_a \le n/(1-\gamma)$. The total flux on paths is at most $n^2$, and the total flux on cycles is at most $n/(1-\gamma)$.
\end{lem}

\begin{proof}
All actions have at least 1 flux. If $a$ is on a path, then starting from any state we can only use $a$ once and never return, contributing flux at most 1 per state, so $\mbx^\pi_a \le n$. Summing over all path actions, the total flux is at most $n^2$.

If $a$ is on a cycle, each state on the cycle contributes a total of $1/(1-\gamma)$ flux to the cycle. By symmetry this flux is distributed evenly among actions on the cycle, so $\mbx^\pi_a \ge 1/(1-\gamma)$. The total flux in the MDP is $n/(1-\gamma)$, so $\mbx^\pi_a \le n/(1-\gamma)$.
\end{proof}

The overall range of flux is large, but all values must lie within one of two polynomial layers. We will prove that simplex can essentially optimize each layer separately. If a cycle is not updated, then not much progress is made towards the optimum, but we make a substantial amount of progress in optimizing the paths for the current cycles. When the paths are optimal the algorithm is forced to update a cycle, at which point we make a substantial amount of progress towards the optimum but resets all progress on the paths.

First we analyze progress on the paths:

\begin{lem}
\label{lem_path_converge}
Suppose the simplex method pivots from $\pi$ to $\pi'$, which does not create a new cycle. Let $\pi''$ be the final policy such that cycles in $\pi''$ are a subset of those in $\pi$ (i.e., the final policy before a new cycle is created).
Then $\mbr^T(\mbx^{\pi''} - \mbx^{\pi'}) \le (1-1/n^2)\mbr^T(\mbx^{\pi''} - \mbx^{\pi})$.
%Suppose the $t+1$th pivot does not create a new cycle. Let $t'$ be the final step such that the cycles in $\pi^{t'}$ are a subset of those in $\pi^t$.  Then $\mbr^T(\mb{x}^{t'} - \mb{x}^{t+1}) \le (1-1/n^2)\mbr^T(\mb{x}^{t'} - \mb{x}^t)$.
\end{lem}

\begin{proof}
Let $\Delta = \max_a \mbr^\pi_a$ be the highest gain. 
Consider $(\mbr^{\pi})^T\mbx^{\pi''}$. 
Since cycles in $\pi''$ are contained in $\pi$, $\mbr^{\pi}_a = 0$ for any action $a$ on a cycle in $\pi''$, and by Lemma \ref{lem_flux_size}, $\pi''$ has at most $n^2$ units of flux on paths, so $(\mbr^{\pi})^T\mbx^{\pi''} = \mbr^T(\mbx^{\pi''} - \mbx^\pi) \le n^2\Delta$.

%For any policy, each of $n$ starting states reaches a cycle after at most $n$ steps, so there are at most $n^2$ total units of flow on the paths. Since cycles in $\pi^{t'}$ are contained in $\pi^t$, $\mbr^{t}_a = 0$ for any action $a$ on a cycle in $\pi^{t'}$, so $\mbr^T\mb{x}^{t'} - \mbr^T\mb{x}^{t} = (\mbr^t)^T\mb{x}^{t'} \le n^2\Delta$.

Policy $\pi'$ has at least 1 unit of flux on the action with gain $\Delta$, so
\[
\mbr^T(\mbx^{\pi''} - \mbx^{\pi'}) \le \mbr^T(\mbx^{\pi''} - \mbx^{\pi}) - \Delta \le \left(1-\frac{1}{n^2}\right)\mbr^T(\mbx^{\pi''} - \mbx^{\pi}) \; . \qedhere
\]
\end{proof}

Due to the polynomial contraction in the lemma above, not too many iterations can pass before a new cycle is formed.

\begin{lem}
\label{lem_path_elimination}
Let $\pi$ be a policy. After $O(n^2\log n)$ iterations starting from $\pi$, either the algorithm finishes, a new cycle is created, a cycle is broken, or some action in $\pi$ never appears in a policy again until a new cycle is created.
\end{lem}

\begin{proof}
Let $\pi$ be the policy in some iteration, $\pi'$ the last policy before a new cycle is created, and $\pi''$ an arbitrary policy occurring between $\pi$ and $\pi'$ in the algorithm.
Policy $\pi$ differs from $\pi'$ in actions on paths and possibly in cycles that exist in $\pi$ but have been broken in $\pi'$. By Lemma \ref{lem_reducedcosts} $-(\mbr^{\pi'})^T\mbx^\pi = \mbr^T(\mbx^{\pi'}- \mbx^{\pi}) = \mb{1}^T(\mb{v}^{\pi'}-\mb{v}^{\pi})$.

We divide the analysis into two cases. First suppose that there exists an action $a$ used in state $s$ on a path such that $-\mbr^{\pi'}_a\mbx^{\pi}_a \ge  -(\mbr^{\pi'})^T\mbx^{\pi}/n$  (note $(\mbr^{\pi'})^T\mbx^\pi \le 0$). Since $a$ is on a path $\mbx^\pi_a \le n$, which implies $-\mbr^{\pi'}_a n^2 \ge -(\mbr^{\pi'})^T\mbx^{\pi}$. Now if policy $\pi''$ uses action $a$, then
\begin{align*}
-(\mbr^{\pi'})^T\mbx^{\pi''} = 
\mb{1}^T(\mb{v}^{\pi'}-\mb{v}^{\pi''}) \ge 
\mb{v}^{\pi'}_s - \mb{v}^{\pi''}_s =
& \mb{v}^{\pi'}_s - (\mbr_a + \gamma P_a \mb{v}^{\pi''}) \\ \ge
& \mb{v}^{\pi'}_s - (\mbr_a + \gamma P_a \mb{v}^{\pi'}) =
-\mbr^{\pi'}_a \ge 
-\frac{-(\mbr^{\pi'})^\pi\mbx^{\pi}}{n^2} \; ,
\end{align*}
using that the values of all states are monotone increasing. 

In the second case there is no action $a$ on a path in $\pi$ satisfying $-\mbr^{\pi'}_a\mbx^{\pi}_a \ge  -(\mbr^{\pi'})^T\mbx^{\pi}/n$. The remaining portion of $-(\mbr^{\pi'})^T\mbx^{\pi}$ is due to cycles, so there must be some cycle $C$ consisting of actions $\{a_1, \ldots, a_k\}$ used in states $\{s_1,\ldots, s_k\}$ such that $\sum_{a \in C} -\mbr^{\pi'}_a \mbx^\pi_a \ge -(\mbr^{\pi'})^T\mbx^{\pi}/n$. 

All flux in $C$ first enters $C$ either from a path ending at $C$ or from the initial unit of flux placed on some state $s$ in $C$. If $y_s \ge 1$ units of flux first enter $C$ at state $s$ in policy $\pi$, then that flux earns $y_s(\mbv^{\pi'}_s-\mbv^{\pi}_s)$ reward with respect to the rewards $-\mbr^{\pi'}$, so $\sum_{a \in C} -\mbr^{\pi'}_a \mbx^\pi_a = \sum_{s \in C} y_s(\mbv^{\pi'}_s-\mbv^{\pi}_s)$. Moreover, each term $\mbv^{\pi'}_s - \mbv^{\pi}_s$ is nonnegative, since the values of all states are nondecreasing. 
Now note that $\sum_{s \in C} (\mbv^{\pi'}_s - \mbv^{\pi}_s) = \sum_{a \in C} -\mbr^{\pi'}_a/(1-\gamma)$, and at most $n$ units of flux enter each state from outside. Therefore $-n \sum_{a \in C} \mbr^{\pi'}_a/(1-\gamma) \ge \sum_{a \in C} -\mbr^{\pi'}_a \mbx^\pi_a$, implying $-n^2 \sum_{a \in C} \mbr^{\pi'}_a/(1-\gamma) \ge -(\mbr^{\pi'})^T\mbx^{\pi}$.

As long as cycle $C$ is intact, each $a \in C$ has $1/(1-\gamma)$ flux from states in $C$ (Lemma \ref{lem_flux_size}), so if $C$ is in policy $\pi''$ then
\begin{equation}
\label{eqn_cycle_breaking}
-(\mbr^{\pi'})^T\mbx^{\pi''} =
\mb{1}^T(\mb{v}^{\pi'}-\mb{v}^{\pi''}) \ge 
\sum_{s \in C} \mb{v}^{\pi'}_s - \mb{v}^{\pi''}_s =
 -\frac{\sum_{a \in C} \mbr^{\pi''}_a}{1-\gamma} \ge
 -\frac{-(\mbr^{\pi'})^T\mbx^{\pi}}{n^2} \; .
\end{equation}

Now if $\log_{n^2/(n^2-1)} n^2$ iterations occur between $\pi$ and $\pi''$, Lemma \ref{lem_path_converge} implies
\[
-(\mbr^{\pi'})^T\mbx^{\pi''} <
-\left(1-\frac{1}{n^2}\right)^{\log_{n^2/(n^2-1)} n^2} (\mbr^{\pi'})^T\mbx^{\pi} \le
 -\frac{-(\mbr^{\pi'})^T\mbx^{\pi}}{n^2} \; .
\]
In the first case action $a$ cannot appear in $\pi''$, and in the second case cycle $C$ must be broken broken in $\pi''$.
This takes $\log_{n^2/(n^2-1)} n^2 = O(n^2\log n)$ iterations if no new cycles interrupt the process.
\end{proof}

\begin{lem}
\label{lem_new_cycle}
Either the algorithm finishes or a new cycle is created after $O(n^2 m \log n)$ iterations.
\end{lem}

\begin{proof}
Let $\pi_0$ be a policy after a new cycle is created, and consider the policies $\pi_1, \pi_2, \ldots$ each separated by $O(n^2\log n)$ iterations. If no new cycle is created, then by Lemma \ref{lem_path_elimination} each of these policies $\pi_i$ has either broken another cycle in $\pi_0$ or contains an action that cannot appear in $\pi_{j}$ for all $j >i$. There are at most $n$ cycles in $\pi_0$ and at most $m$ actions that can be eliminated, so after $(m+n)O(n^2 \log n) = O(n^2 m \log n)$ iteration, the algorithm must terminate or create a new cycle.
\end{proof}

When a new cycle is formed, the algorithm makes a substantial amount of progress towards the optimum but also resets the path optimality above.

\begin{lem}
\label{lem_cycle_converge}
Let $\pi$ and $\pi'$ be subsequent policies such that $\pi'$ creates a new cycle.
Then $\mbr^T(\mbx^* - \mbx^{\pi'}) \le (1-1/n)\mbr^T(\mbx^* - \mbx^\pi)$.
\end{lem}

\begin{proof}
Let $\Delta = \max_{a'} \mbr^\pi_{a'}$ and $a = \argmax_{a'} \mbr^\pi_a$.  There is a total of $n/(1-\gamma)$ flux in the MDP, so $\mbr^T\mbx^* - \mbr^T\mbx^{\pi} = (\mbr^\pi)^T\mbx^* \le \Delta n/(1-\gamma)$. 
By Lemma \ref{lem_flux_size}, pivoting on $a$ and creating a cycle will result in at least $1/(1-\gamma)$ flux through $a$.
Therefore $\mbr^T\mbx^{\pi'} \ge \mbr^T\mbx^\pi + \Delta/(1-\gamma)$, so
\[
\mbr^T(\mbx^* - \mbx^{\pi'}) \le \mbr^T(\mbx^* - \mbx^{\pi}) - \frac{\Delta}{1-\gamma} \le \left(1-\frac{1}{n}\right)\mbr^T(\mbx^* - \mbx^{\pi}) \; . \qedhere
\]
\end{proof}

\begin{lem}
\label{lem_cycle_elimination}
Let $\pi$ be a policy. Starting from $\pi$, after $O(n \log n)$ iterations in which a new cycle is created, some action in $\pi$ is either eliminated from cycles for the remainder of the algorithm or entirely eliminated from policies for the remainder of the algorithm.
\end{lem}

\begin{proof}
Consider a policy $\pi$ with respect to the optimal gains $\mbr^*$. There is an action $a$ such that $-\mbr^*_a\mbx^\pi_a \ge -(\mbr^*)^T\mbx^\pi/n$. If $a$ is on a path in $\pi$, then $1 \le \mbx^\pi_a \le n$, so $-\mbr^*_a \ge -(\mbr^*)^T\mbx^\pi/n^2$, and if $a$ is on a cycle, then $1/(1-\gamma) \le \mbx^\pi_a \le n/(1-\gamma)$, so $-\mbr^*_a/(1-\gamma) \ge -(\mbr^*)^T\mbx^\pi/n^2$.

Since $\mbr^*$ are the gains for the optimal policy, $\mbr^*_{a'} \le 0$ for all $a'$. Therefore if $\pi'$ is any policy containing $a$, then 
$-\mbr^*_a \le -\mbr^*_a \mbx^{\pi'}_a \le -(\mbr^*)^T\mbx^{\pi'}$, 
and if $\pi'$ is any policy containing $a$ on a cycle, then
$-\mbr^*_a/(1-\gamma) \le -\mbr^*_a \mbx^{\pi'}_a \le -(\mbr^*)^T\mbx^{\pi'}$.
Now by Lemma \ref{lem_cycle_converge}, if there are more than $\log_{n/(n-1)} n^2 = O(n \log n)$ new cycles created between policies $\pi$ and $\pi'$ then
\[
-(\mbr^*)^T\mbx^{\pi'} < -\left(1-\frac{1}{n}\right)^{\log_{n/(n-1)} n^2}(\mbr^*)^T\mbx^{\pi} = -\frac{(\mbr^*)^T\mbx^\pi}{n^2} \; .
\]
Therefore if $\pi$ contained $a$ on a path, then $a$ cannot appear in any policy after $\pi'$ for the remainder of the algorithm, and if $\pi$ contained $a$ on a cycle, then $a$ cannot appear in a cycle after $\pi'$ (but may appear in a path) for the remainder of the algorithm.
\end{proof}

\begin{thm}
The simplex method converges in at most $O(n^3 m^2 \log^2 n)$ iterations on deterministic MDPs with uniform discount using the highest gain pivoting rule.
\end{thm}

\begin{proof}
Consider the policies $\pi_0, \pi_1, \pi_2,\ldots$ where $O(n \log n)$ new cycles have been created between $\pi_i$ and $\pi_{i+1}$. 
By Lemma \ref{lem_cycle_elimination}, each $\pi_i$ contains an action that is either eliminated entirely in $\pi_j$ for $j>i$ or eliminated from cycles. Each action can be eliminated from cycles and paths, so after $2m$ such rounds of $O(n\log n)$ new cycles the algorithm has converged. By Lemma \ref{lem_new_cycle} cycles are created every $O(n^2 m \log n)$ iterations, for a total of $O(n^3 m^2 \log^2 n)$ iterations.
\end{proof}

\section{Varying Discounts}

\label{section_nonuniform}

In this section we allow each action $a$ to have a distinct discount $\gamma_a$. This significantly complicates the proof of convergence since the total flux is no longer fixed. When updating a cycle we can no longer bound the distance to the optimum based solely on the maximum gain, since the optimal policy may employ actions with smaller gain to the current policy but substantially more flux.

We are able to exhibit a set of layers in which the flux on cycles must lie based on the discount of the actions, and we  will show that when a cycle is created in a particular layer we make progress towards the optimum value for the updated state {\em assuming that it lies within that layer}. These layers will define a set of bounds whose values we must surpass, which serve as milestones or checkpoints to the optimum.
When we update a cycle we cannot claim that the overall objective increases substantially but only that the values of individual states make progress towards one of these milestone values. When the values of all states have surpassed each of these intermediate milestones the algorithm will terminate.

We first define some notation. Recall that to calculate flux we place one unit of ``mass'' in each state and then run the Markov chain, so all flux traces back to some state, but $\mbx^\pi$ aggregates all of it together. Because we will be concerned with analyzing the values of individual states in this section, it will be useful to separate out the flux originating in a particular state $s$. Consider the following alternate LP:

\begin{equation}
\label{eq_mod_primal}
\begin{array}{rlrl}
\textrm{maximize} &  \mbr^T\mbx  \\
\textrm{subject to} & & \sum_{a \in A_s} \mbx_a & = 1 + \sum_{a} \gamma_a P_{a,s}\mbx_a\\
& \forall s' \neq s & \sum_{a \in A_{s'}} \mbx_a  & =  \sum_{a} \gamma_a P_{a,s'} \mbx_a\\
& & \mbx & \ge 0 \\
\end{array}
\end{equation}

The LP \eqref{eq_mod_primal} is identical to \eqref{eq_primal}, except that initial flux is only added to state $s$ rather than all states, and the dual of \eqref{eq_mod_primal} matches \eqref{eq_dual} if the objective in \eqref{eq_dual} is changed to minimize only $\mbv_s$. Feasible solutions in \eqref{eq_mod_primal} measure only flux originating in $s$ and contributing to $\mbv_s$. For a state $s$ and policy $\pi$ we use the notation $\mbx^{\pi,s}$ to denote the corresponding vertex in \eqref{eq_mod_primal}. Note that $\mbx^\pi = \sum_s \mbx^{\pi,s}$.

The following lemma is analogous to Lemma \ref{lem_reducedcosts} and has an identical proof:

\begin{lem}
\label{lem_subset_reducedcosts}
For a state $s$ and for policies $\pi$ and $\pi'$, $(\mb{r}^\pi)^T\mb{x}^{\pi',s} = \mb{r}^T\mb{x}^{\pi',S} - \mb{r}^T\mb{x}^{\pi,S} = \mb{v}^{\pi'}_s - \mb{v}^\pi_s$.
\end{lem}

We now define the intervals in which the flux must lie. As in Section \ref{section_uniform} flux on paths is in $[1,n]$. Let $C$ be a cycle in some policy, and $\gamma_C = \prod_{a\in C} \gamma_a$ be total discount of $C$. We will prove that the smallest discount in $C$ determines the rough order of magnitude of the flux through $C$.

\begin{defn} 
Let $C$ be a cycle and $a$ an action in $C$, then the discount of $a$ {\em dominates} the discount of $C$ if $\gamma_a \le \gamma_{a'}$ for all $a' \in C$.
\end{defn}

\begin{lem} 
\label{lem_cycle_flow_nonuniform}
Let $\pi$ be a policy containing the cycle $C$ with discount dominated by $\gamma_a$ and total discount $\gamma_C$.  Let $s$ be a state on $C$, $a'$ the action used in $s$ and $a''$ an arbitrary action in $C$, then 
\begin{itemize}
\item $\mb{x}^{\pi,s}_{a'} = 1/(1-\gamma_C)$,
\item $\gamma_C/(1-\gamma_C) \le \mb{x}^{\pi,s}_{a''} \le 1/(1-\gamma_C)$, and
\item $1/(n(1-\gamma_a)) \le 1/(1-\gamma_C) \le 1/(1-\gamma_a)$.
\end{itemize}
\end{lem}

\begin{proof}
For the first equality, all flux originates at $s$, so the flux through $a'$ (used in state $s$) either just originated in $s$ or came around the cycle from $s$, implying $\mb{x}^{\pi,s}_{a'} = 1 + \gamma_C\mb{x}^{\pi,s}_{a'}$.
An analogous equation holds for all other actions $a''$ on $C$, but now the initial flow from $s$ may have been discounted by at most $\gamma_C$ before reaching $a''$, giving $\gamma_C/(1-\gamma_C) \le \mb{x}^{\pi,s}_{a''} \le 1/(1-\gamma_C)$.

The upper bound in the final inequality, $1/(1-\gamma_C) \le 1/(1-\gamma_a)$ holds since $a \in C$ ($\gamma_a$ dominates the discount of $C$). For the lower bound, let $\ell = 1-\gamma_a$. Then $\gamma_C \ge \gamma_a^n = (1-\ell)^n \ge 1-n\ell = 1-n(1-\gamma_a)$, implying $1/(1-\gamma_C) \ge 1/(n(1-\gamma_a))$.
\end{proof}

Flux on paths still falls in $[1,n]$, so the algorithm behaves the same on paths as it did in the uniform case:

\begin{lem}
\label{lem_new_cycle_nonuniform}
Either the algorithm finishes or a new cycle is created after $O(n^2 m \log n)$ iterations.
\end{lem}

\begin{proof}
This is identical to the proof of Lemma \ref{lem_new_cycle}, which depends on Lemmas \ref{lem_path_converge} and \ref{lem_path_elimination}. Lemma \ref{lem_path_converge} holds for nonuniform discounts, and Lemma \ref{lem_path_elimination} holds after adjusting Equation \eqref{eqn_cycle_breaking} as follows
\[
-(\mb{r}^{\pi'})^T\mb{x}^{\pi''} \ge 
\sum_{s \in C} \mb{v}^{\pi'}_s - \mb{v}^{\pi''}_s \ge
 -\frac{\sum_{a \in C} \mb{r}^{\pi''}_a}{1-\gamma_C} \ge
 -\frac{-(\mb{r}^{\pi'})^T\mb{x}^{\pi}}{n^2} \; ,
\]
using that $\sum_{a \in C} \mb{r}^{\pi''}_a n/(1-\gamma_C) \ge -(\mb{r}^{\pi'})^T\mb{x}^{\pi}/n$ and Lemma \ref{lem_cycle_flow_nonuniform}.
\end{proof}

Now suppose the simplex method updates the action for state $s$ in policy $\pi$ and creates a cycle dominated by $\gamma_a$. Again, $\mbv_s$ may not improve much, since there may be a cycle with discount much larger than $\gamma_a$. However, in any policy $\pi'$ where $s$ is on a cycle dominated by $\gamma_{a}$ and $s$ uses some action $a'$, $1/(n(1-\gamma_a)) \le \mbx^{\pi',s}_{a'} \le 1/(1-\gamma_a)$, which allows us to argue $\mbv_s$ has made progress towards the highest value achievable when it is on a cycle dominated by $\gamma_a$, and after enough such progress has made, $\mbv_s$ will beat this value and never again appear on any cycle dominated by $\gamma_a$. The optimal values achievable for each state on a cycle dominated by each $\gamma_a$ serve as the above-mentioned milestones. Since all cycles are dominated by some $\gamma_a$, there are $m$ milestones per state. %the algorithm will find the optimal cycles before passing all $nm$ milestones.

\begin{lem}
\label{lem_cycle_converge_nonuniform}
Suppose the simplex method moves from $\pi$ to $\pi'$ by updating the action for state $s$, creating a new cycle $C$ with discount dominated by $\gamma_a$ for some $a$ in $\pi'$. 
Let $\pi''$ be the final policy used by the simplex method in which $s$ is in a cycle dominated by $\gamma_a$.
Then $\mb{v}_s^{\pi''} - \mb{v}_s^{\pi'} \le (1-1/n^2)(\mb{v}_s^{\pi''} - \mb{v}_s^{\pi})$.
\end{lem}

\begin{proof}
Let $\Delta = \max_{a'} \mb{r}^{\pi}_{a'}$ be the value of the highest gain with respect to $\pi$. Any cycle contains at most $n$ actions, each of which has gain at most $\Delta$ in $\mbr^\pi$, so if $s$ is on a cycle dominated by $\gamma_a$ in $\pi''$ then by Lemma \ref{lem_cycle_flow_nonuniform} and Lemma \ref{lem_subset_reducedcosts},
$\mb{v}_s^{\pi''} - \mb{v}_s^\pi \le n\Delta/(1-\gamma_a)$, and since $\pi'$ creates a cycle dominated by $\gamma_a$, by the same lemmas $\mb{v}_s^{\pi'} \ge \mb{v}_s^\pi +\Delta/(n(1-\gamma_a))$. Combining the two,
\[
\mb{v}_s^{\pi''} - \mb{v}_s^{\pi'} = 
(\mb{v}_s^{\pi''} - \mb{v}_s^{\pi}) - (\mb{v}_s^{\pi'} - \mb{v}_s^{\pi}) \le 
(\mb{v}_s^{\pi''} - \mb{v}_s^{\pi}) - \frac{\Delta}{n(1-\gamma_a)} \le
\left(1-\frac{1}{n^2}\right)(\mb{v}_s^{\pi''} - \mb{v}_s^{\pi}) \; . \qedhere
\]
\end{proof}

The following lemma is the crux of our analysis and allows us to eliminate actions when we get close to a milestone value. This occurs because the positive gains must shrink or else the algorithm would surpass the milestone, and as the positive gains shrink they can no longer balance larger negative gains, forcing such actions out of the cycle.

\begin{lem}
\label{lem_cycle_elimination_nonuniform}
Suppose policy $\pi$ contains a cycle $C$ with discount dominated by $\gamma_a$ and $s$ is a state in $C$. There is some action $a'$ in $C$ (depending on $s$) such that after $O(n^2 \log n)$ iterations that change the action for $s$ and create a cycle with discount dominated by $\gamma_a$, action $a'$ will never again appear in a cycle dominated by $\gamma_a$.
\end{lem}

\begin{proof}
Let $\pi$ be a policy containing a cycle $C$ with discount dominated by $\gamma_a$ and $s$ a state in $C$.  Let $\pi'$ be another policy where $s$ is on a cycle dominated by $\gamma_a$ after at least $1+\log_{n^2/(n^2-1)} n^5 = O(n^2 \log n)$ iterations that create such a cycle by changing the action for $s$ and $\pi''$ the final policy used by the algorithm in which $s$ is on a cycle dominated by $\gamma_a$.

Consider the policy $\hat{\pi}$ in the iteration immediately preceding $\pi'$.
By Lemma \ref{lem_cycle_converge_nonuniform}, and the choice of $\pi'$,
\[
\mb{v}^{\pi''}_s - \mb{v}^{\hat{\pi}}_s \le \left(1-\frac{1}{n^2}\right)^{\log_{n^2/(n^2-1)} n^5}(\mb{v}^{\pi''}_s - \mb{v}^{\pi}_s) = \frac{1}{n^5}(\mb{v}^{\pi''}_s - \mb{v}^{\pi}_s) \; ,
\]
or equivalently $\mb{v}^{\pi}_s - \mb{v}^{\pi''}_s \le -n^5(\mb{v}^{\pi''}_s - \mb{v}^{\hat{\pi}}_s )$, implying
\begin{equation}
\label{eqn_v_gap}
\mb{v}^{\pi}_s - \mb{v}^{\hat{\pi}}_s = (\mb{v}^{\pi}_s - \mb{v}^{\pi''}_s) + (\mb{v}^{\pi''}_s - \mb{v}^{\hat{\pi}}_s) \le (-n^5 +1)(\mb{v}^{\pi''}_s - \mb{v}^{\hat{\pi}}_s) \; .
\end{equation}
Since the gap $\mb{v}^{\pi}_s - \mb{v}^{\hat{\pi}}_s$ is large and negative, there must be highly negative gains in $\mbr^{\hat{\pi}}$.
By Lemma \ref{lem_subset_reducedcosts} $\mb{v}_s^{\pi}-\mb{v}_s^{\hat{\pi}} = (\mb{r}^{\hat{\pi}})^T\mb{x}^{\pi,s}$.
Let $\mb{r}^{\hat{\pi}}_{a'} = \min_{a \in C} \mb{r}^{\hat{\pi}}_a$ and $s'$ be the state using $a'$. By Lemma \ref{lem_cycle_flow_nonuniform}, $\mb{x}^{\pi,s} \le 1/(1-\gamma_a)$, and $C$ has at most $n$ states, so applying Equation \eqref{eqn_v_gap}
\begin{equation}
\label{eqn_min_gain}
\frac{\mb{r}^{\hat{\pi}}_{a'}}{1-\gamma_a} \le \frac{1}{n}(\mb{v}_s^\pi-\mb{v}_s^{\hat{\pi}}) \le \left(-n^4 + \frac{1}{n}\right)(\mb{v}^{\pi''}_s - \mb{v}^{\hat{\pi}}_s) \; .
\end{equation}

The positive entries in $\mbr^{\hat{\pi}}$ must all be small, since there is only a small increase in the value of $s$.
Let $\Delta = \max \mb{r}^{\hat{\pi}}$. The algorithm pivots on the highest gain, and by assumption it updates the action for $s$ and creates a cycle dominated by $\gamma_a$. By Lemma \ref{lem_cycle_flow_nonuniform}, the new action is used at least $1/(n(1-\gamma_a))$ times by flux from $s$, since it is the first action in the cycle, so
\begin{equation}
\label{eqn_max_gain}
\frac{\Delta}{n(1-\gamma_a)} \le \mb{v}_s^{\pi'} - \mb{v}_s^{\hat{\pi}} \le  \mb{v}_s^{\pi''} - \mb{v}_s^{\hat{\pi}} \; .
\end{equation}

We prove that the highly negative $\mbr^{\hat{\pi}}_{a'}$ cannot coexist with only small positive gains bounded by $\Delta$.
Consider any policy in which $s'$ is on a cycle $C'$ containing $a'$ (but not necessarily containing $s$) with total gain $\gamma_{C'}$ dominated by $\gamma_a$. By Lemma \ref{lem_cycle_flow_nonuniform}, there is at least $1/(1-\gamma_{C'}) \ge 1/(n(1-\gamma_a))$ flux from $s$ going through $a'$, and in the rest of the cycle there are at most $n-1$ other actions with at most $1/(1-\gamma_{C'}) \le 1/(1-\gamma_a)$ flux. The highest gain with respect to $\hat{\pi}$ is $\Delta$, so the value of $\mb{v}_{s'}$ relative to $\mb{r}^{\hat{\pi}}$ is at most
\begin{align*}
\frac{\mb{r}^{\hat{\pi}}_{a'}}{n(1-\gamma_a)} + \frac{n\Delta}{1-\gamma_a} 
& \le \left(-n^3 + \frac{1}{n^2}\right)(\mb{v}_s^{\pi''} -\mb{v}_s^{\hat{\pi}}) + n^2(\mb{v}_s^{\pi''} - \mb{v}_s^{\hat{\pi}})  \\
& = \left(-n^3 + \frac{1}{n^2} + n^2\right)(\mb{v}_s^{\pi''} - \mb{v}_s^{\hat{\pi}}) < 0
\end{align*}
using Equations \eqref{eqn_min_gain} and \eqref{eqn_max_gain}.
But $\mb{v}^{\hat{\pi}}_{s'} = 0$ relative to $\mb{r}^{\hat{\pi}}$, and it only increases in future iterations, so $a'$ cannot appear again in a cycle dominated by $\gamma_a$.
\end{proof}

\begin{lem}
\label{lem_discount_elimination}
For any action $a$, there are at most $O(n^3m \log n)$ iterations that create a cycle with discount dominated by $\gamma_a$.
\end{lem}

\begin{proof}
After $O(n^3 \log n)$ iterations that create a cycle dominated by $\gamma_a$, some state must have been updated in $O(n^2 \log n)$ of those iterations, so by Lemma \ref{lem_cycle_elimination_nonuniform} some action will never appear again in a cycle dominated by $\gamma_a$. After $m$ repetitions of this process all actions have been eliminated.
\end{proof}

\begin{thm}
Simplex terminates in at most $O(n^5 m^3 \log^2 n)$ iterations on deterministic MDPs with nonuniform discounts using the highest gain pivoting rule.
\end{thm}

\begin{proof}
There are $O(m)$ possible discounts $\gamma_a$ that can dominate a cycle, and by Lemma \ref{lem_discount_elimination} there are at most $O(n^3 m \log n)$ iterations creating a cycle dominated by any particular $\gamma_a$, for a total of $O(n^3 m^2 \log n)$ iterations that create a cycle. By Lemma \ref{lem_new_cycle_nonuniform} a new cycle is created every $O(n^2 m \log n)$ iterations, for a total of $O(n^5 m^3 \log^2 n)$ iterations overall.
\end{proof}

\section{Open problems}

A difficult but natural next step would be to try to extend these techniques to handle policy iteration on deterministic MDPs. The main problem encountered is that the multiple simultaneous pivots used in policy iteration can interfere with each other in such a way that the algorithm effectively pivots on the {\em smallest} improving switch rather than the largest. See \cite{hansenzwick_mmc} for such an example. Another challenging open question is to design a strongly polynomial algorithm for general MDPs. Finally, we believe the technique of dividing variable values into polynomial sized layers may be helpful for entirely different problems.

\pdfbookmark[1]{\refname}{My\refname}
\bibliographystyle{alphaurl}
\bibliography{mdp_simplex}

\end{document}